\documentclass[conference]{IEEEtran}
\usepackage[T1]{fontenc} 
\usepackage{times}
\usepackage{amsmath}
\usepackage{amssymb}
\usepackage{amsthm}
\usepackage{color}
\usepackage{multicol}
\usepackage{colortbl}
\usepackage{booktabs}
\usepackage{array}
\usepackage{algorithm}
\usepackage[noend]{algorithmic}
\usepackage{float}
\usepackage{graphicx}
\usepackage{subfig}
\usepackage{array}
\usepackage{multirow}
\usepackage[bookmarks=false,colorlinks=false,pdfborder={0 0 0}]{hyperref}
\usepackage{cite}
\usepackage{bm}
\usepackage{arydshln}
\usepackage{mathtools}
\usepackage{microtype}

\usepackage[figuresright]{rotating}
\usepackage{threeparttable}
\usepackage{booktabs}

\newtheorem{theorem}{Theorem}

\newtheorem{lemma}[theorem]{Lemma}

\newtheorem{example}[theorem]{Example}
\long\def\symbolfootnote[#1]#2{\begingroup
\def\thefootnote{\fnsymbol{footnote}}\footnote[#1]{#2}\endgroup}
\renewcommand{\paragraph}[1]{{\bf #1}}

\title{Set Transformation: Trade-off Between Repair Bandwidth and Sub-packetization}

\author{Hao Shi$^\dagger$, Zhengyi Jiang$^\dagger$,  Zhongyi Huang$^\dagger$,
Bo Bai$^\ddagger$, Gong Zhang$^\ddagger$, and Hanxu Hou$^\ddagger$$^{\star}$\\
$^\dagger$ Department of Mathematics Sciences, Tsinghua University, Beijing, China \\
$^\ddagger$ Theory Lab, Central Research Institute, 2012 Labs, Huawei Tech. Co. Ltd., Hong Kong SAR
}

\begin{document}
\let\emph\textit
\maketitle
\pagestyle{empty}  
\thispagestyle{empty} 
\symbolfootnote[0]{
$^{\star}$: Corresponding author.
}
\begin{abstract}
Maximum distance separable (MDS) codes facilitate the achievement of elevated levels of fault tolerance in storage systems while incurring minimal redundancy overhead. 
Reed-Solomon (RS) codes are typical MDS codes with the sub-packetization level being one, however, they require large repair bandwidth defined as the total amount of symbols downloaded from other surviving nodes during single-node failure/repair.
In this paper, we present the  {\em set transformation}, which can transform any MDS code into set transformed code such that (i) the sub-packetization level is flexible and ranges from 2 to $(n-k)^{\lfloor\frac{n}{n-k}\rfloor}$ in which $n$ is the number of nodes and $k$ is the number of data nodes, (ii) the new code is MDS code, (iii) the new code has lower repair bandwidth for any single-node failure.
We show that our set transformed codes have both lower repair bandwidth and lower field size than the existing related MDS array codes, such as elastic transformed codes \cite{10228984}. Specifically, our set transformed codes have $2\%-6.6\%$ repair bandwidth reduction compared with elastic transformed codes \cite{10228984} for the evaluated typical parameters.


\end{abstract}

\IEEEpeerreviewmaketitle

\section{Introduction}


Maximum distance separable (MDS) codes are widely employed in storage systems to provide optimal trade-off between storage overhead and fault tolerance. 
An $(n,k,\alpha)$ MDS array code encodes $k\alpha$
{\em data symbols} into $n\alpha$ {\em coded symbols} that are equally stored in $n$ nodes,
where each node stores $\alpha$ symbols. We call the number of symbols stored in each node
as the sub-packetization level. The $(n,k,\alpha)$ MDS array codes satisfy the {\em MDS property},
that is, any $k$ out of $n$ nodes can retrieve all $k\alpha$ data symbols. The codes are called {\em systematic codes} if the $k\alpha$ data symbols are included in the $n\alpha$ coded symbols. Since non-systematic codes can be converted into systematic codes through linear transformation, we will not make a distinction in this paper.
Reed-Solomon (RS) codes \cite{reed1960} are typical MDS array codes with the sub-packetization level $\alpha = 1$. 

Repair bandwidth is an important metric since node failure is common in distributed storage systems. 
It is shown in \cite{ford2010} that single-node failure occurs 
more frequently among all failures. It is important to
repair the failed node with the repair bandwidth as small as possible \cite{dimakis2010}.
Dimakis \emph{et al.} showed that we can repair the failed node by accessing at least 
$\frac{\alpha}{n-k}$ symbols from each of the other $n-1$ surviving nodes and the MDS array 
codes achieving the minimum repair bandwidth $\frac{(n-1)\alpha}{n-k}$ are called {\em minimum storage regenerating} 
(MSR) codes. Many MSR codes have been proposed \cite{rashmi2011,tamo2013,hou2016,2017Explicit,li2018,2018A,hou2019a,hou2019b}. 
However, high-code-rate (i.e., $\frac{k}{n}>\frac{1}{2}$) MSR codes \cite{2018A} require that the sub-packetization 
level increases exponentially with parameters $n$ and $k$.
It is practically important to design high-code-rate MDS array codes with repair bandwidth
as small as possible, for a given small sub-packetization level.

The piggybacking framework  \cite{2017Piggybacking} proposed by Rashmi \emph{et al.} can generate MDS array codes with low repair bandwidth and low sub-packetization level over a small field size. Many follow-up piggybacking codes were proposed \cite{2019AnEfficient,2021piggyback,2022Piggyback,2023jzy, 10495316}. 
However, how the piggybacking framework used to maintain the MDS property limits the repair bandwidth reduction of piggybacking codes.

To further reduce the repair bandwidth, many MDS array codes with richer structures have been proposed \cite{2023gcplus,8025778,8006804,10228984,2023wk1,2023wk2}.
Shi \emph{et al.} \cite{2023gcplus} proposed piggybacking+ codes via efficient transformation to reach lower repair bandwidth than piggybacking codes over a small field size and linear sub-packetization level.
HashTag Erasure Codes (HTEC) \cite{8025778} have flexible sub-packetization $2\leq \alpha\leq r^{\left \lceil \frac{k}{r} \right \rceil }$ and low repair bandwidth.
However, the construction of HTEC is not explicit and HTEC only considered efficient repair for data nodes. Moreover, the field size of HTEC is large ($O(\binom{n}{k}(n-k)\alpha)$) for maintaining MDS property.
Wang \emph{et al.} \cite{2023wk1,2023wk2} proposed Bidirectional Piggybacking Design (BPD) with small sub-packetization $2\leq \alpha\leq r$ to reduce the repair bandwidth for each node over a relatively small field size.

Li \emph{et al.} proposed the {\em base transformation} \cite{8006804} that can convert a general MDS code called base code into a new MDS code with reduced repair bandwidth by increasing sub-packetization level without altering the finite field size. 
However, base transformation \cite{8006804} has two limitations: (i) it is not suitable for binary MDS codes; (ii) the sub-packetization level of the transformed code should be $r$ times that of the base code. 
Recently, elastic transformation proposed in \cite{10228984} can generate elastic transformed MDS codes with flexible sub-packetization level $2\leq \alpha\leq r^{\lfloor\frac{n}{r}\rfloor}$ to achieve small repair bandwidth for each node. To maintain the MDS property, the field size is increased to $O(2(\binom{n - 1}{k - 1} - \binom{\lceil \frac{n}{\alpha*} \rceil - 1}{\lceil \frac{k}{\alpha*} \rceil - 1}))$, where $2 \leq \alpha* \leq r$\cite{10228984}.


In this paper, we present a new transformation called {\em set transformation} that can generate set transformed MDS codes with lower repair bandwidth than the existing related codes.
Compared with elastic transformation, our set transformation has two advantages. First, our set transformed codes have lower repair bandwidth than elastic transformed codes. Second, the field size of our set transformed codes is lower than that of elastic transformed codes.
The essential reason for our set transformed codes obtaining the above two advantages is as follows.
Elastic transformation converts a non-square array into a square array by performing base transformation \cite{8006804} twice. While in our set 
transformation, we first skillfully divide the non-square array into some square arrays, and then employ the base transformation \cite{8006804} for each square array. 

\section{Set Transformation}

In this section, we present the construction of set transformation, which can be used to generate MDS array codes with lower repair bandwidth.

Our set transformation can transform any $\alpha \times \beta$ array into another $\alpha \times \beta$ array, where $\alpha$ and $\beta$ are positive integers with  $\alpha \leq \beta$. For easier presentation, in the following, we assume that $\alpha \leq \beta < 2\alpha$.
Denote these $\alpha\beta$ symbols of the $\alpha \times \beta$ array as $\{b_{i,j}\}_{i = 1,2,\ldots,\alpha}^{j= 1,2,\ldots,\beta}$. For any $i\in\{1,2,\ldots,\alpha\}$ and $j\in\{1,2,\ldots,\beta\}$, $b_{i,j}$ represents the symbol in the $i$-th row and $j$-th column.


The construction of general set transformation is divided into the following three steps.
\begin{itemize}
\item \textbf{ Step one (Sub-array Allocation):} Divide the $\alpha \times \beta$ array into the $(2\alpha - \beta) \times (2\alpha - \beta)$ sub-array $A$, $(2\alpha - \beta) \times (2\beta - 2\alpha)$ sub-array $B_1$, $(\beta - \alpha) \times (2\alpha - \beta)$ sub-array $B_2$ and $(\beta - \alpha) \times (2\beta - 2\alpha)$ sub-array $C$. Fig. \ref{ABC} shows the sub-array allocation. Although we can divide the array into some other sub-arrays, the repair bandwidth of the codes with the above allocation is minimum among all the allocations.
\item \textbf{ Step two (Set Allocation):}  Divide the $\alpha \beta$ symbols in the $\alpha \times \beta$ array into  $\alpha^2$ sets $R_{i,j}$, where $i,j\in \{1,2,\ldots,\alpha\}$. Specifically, let $R_{i,j}$ be
   \begin{equation}
        \begin{cases}
        \{b_{i,j}\} & 
    1 \leq j \leq 2\alpha -\beta; \\
        \{b_{i,2j-2\alpha+\beta - 1}, b_{i,2j-2\alpha+\beta}\} &2\alpha - \beta + 1\leq j \leq \alpha,
            \end{cases}
    \end{equation}
where $1 \leq i \leq \alpha$.
There are $\alpha\times (2\alpha - \beta)$ sets containing one symbol, and the remaining $\alpha\times (\beta-\alpha)$ sets containing two symbols.
Therefore, $A$ contains $(2\alpha-\beta)^2$ sets $\{R_{i,j}\}_{i=1,2,\ldots,2\alpha-\beta}^{j=1,2,\ldots,2\alpha-\beta}$; $B_1$ contains $(2\alpha-\beta)(\beta-\alpha)$ sets $\{R_{i,j}\}_{i=1,2,\ldots,2\alpha-\beta}^{j=2\alpha - \beta + 1, \ldots, \alpha}$; $B_2$ contains $(2\alpha-\beta)(\beta-\alpha)$ sets $\{R_{i,j}\}_{i=2\alpha -\beta + 1,\ldots, \alpha}^{j=1,2,\ldots,2\alpha-\beta}$; $C$ contains $(\beta-\alpha)^2$ sets $\{R_{i,j}\}_{i=2\alpha-\beta + 1,\ldots,\alpha}^{j=2\alpha-\beta + 1,\ldots,\alpha}$. 
\item \textbf{ Step three (Set Pairwise Combination):} For $i\neq j\in\{ 1,2,\ldots,\alpha\}$, define $R_{i,j}$ and $R_{j,i}$ as a pair of {\em coupled sets}.
We perform linear combinations for all the symbols in the coupled sets $R_{i,j}$ and $R_{j,i}$ for $i\neq j$ to update the symbols to be $\{b'_{i,j}\}_{i=1,2,\ldots,\alpha}^{j=1,2,\ldots,\beta}$ and the sets to be $\{R'_{i,j}\}_{i,j=1,2,\ldots,\alpha}$. The symbols in the sets $R_{i,i}$ for $i=1,2,\ldots,\alpha$ are unchanged. For $i\neq j\in\{ 1,2,\ldots,\alpha\}$, we call the symbols in $R'_{i,j}$ and $R'_{j,i}$ as coupled symbols. 
Specifically, the updated $(2\alpha-\beta)^2$ sets in $A$ are     
\begin{equation}
        R'_{i,j} = \begin{cases}
        \{b'_{i,j} = b_{i,j} + b_{j,i}\} &i < j; \\
        \{b'_{i,j} = b_{i,j}\} &i=j; \\	
        \{b'_{i,j} = b_{i,j} + \theta_{i,j}\cdot b_{j,i}\} & i > j,
            \end{cases}
    \end{equation}
    where $i,j = 1,2,\ldots,2\alpha-\beta$ and $\theta_{i,j} \in \mathbb{F}_q\backslash\{0,1\}$.
The $(2\alpha-\beta)(\beta-\alpha)$ sets in $B_1$ are updated to be  
    \begin{eqnarray*}
        R'_{i,j} &&= \{b'_{i,2j - 2\alpha +\beta-1} = b_{i,2j - 2\alpha +\beta-1} + b_{j,i}, \\&&b'_{i,2j - 2\alpha + \beta} = b_{i,2j - 2\alpha + \beta}\},
    \end{eqnarray*}
    where $i = 1,2,\ldots,2\alpha - \beta$ and $j = 2\alpha-\beta+1,\ldots,\beta$.
The $(2\alpha-\beta)(\beta-\alpha)$ sets in $B_2$ are updated to be  
\[R'_{i,j} =\{b'_{i,j}=b_{i,j}+\theta_{i,j}\cdot(b_{j,2i-2\alpha+\beta-1}+b_{j,2i-2\alpha+\beta})\},\]
    where $i = 2\alpha-\beta+1,\ldots,\alpha, j = 1,2,\ldots,2\alpha-\beta$ and $\theta_{i,j} \in \mathbb{F}_q\backslash\{0,1\}$.
   The $(\beta-\alpha)^2$ sets $R'_{i,j}$ in $C$ are updated to be     

    \begin{align*}
        \begin{cases}
            \{b'_{i,2j - 2\alpha +\beta-1} = \\b_{i,2j - 2\alpha +\beta-1} + \theta_{i,2j - 2\alpha +\beta-1}\cdot  b_{j,2i -2\alpha + \beta - 1},\\
            b'_{i,2j - 2\alpha +\beta} =\\ b_{i,2j - 2\alpha +\beta} + \theta_{i,2j - 2\alpha +\beta}\cdot  b_{j,2i -2\alpha + \beta}\},\  i > j\\
            \{b'_{i,2j-2\alpha+\beta-1} = b_{i,2j-2\alpha+\beta-1}, \\b'_{i,2j - 2\alpha+\beta}=b_{i,2j - 2\alpha+\beta}\},\ i=j \\
            \{b'_{i,2j - 2\alpha +\beta-1} = b_{i,2j - 2\alpha +\beta-1} + b_{j,2i -2\alpha + \beta - 1},\\
            b'_{i,2j - 2\alpha +\beta} = b_{i,2j - 2\alpha +\beta} + b_{j,2i -2\alpha + \beta}\},\ i < j
        \end{cases}
    \end{align*}
    where $i = 2\alpha-\beta+1,\ldots,\alpha$, $j = 2\alpha-\beta+1,\ldots, \beta$ and $\theta_{i,j} \in \mathbb{F}_q\backslash\{0,1\}$.
\end{itemize}

\begin{figure}[H]
    \centering
    \includegraphics[width=0.42\textwidth]{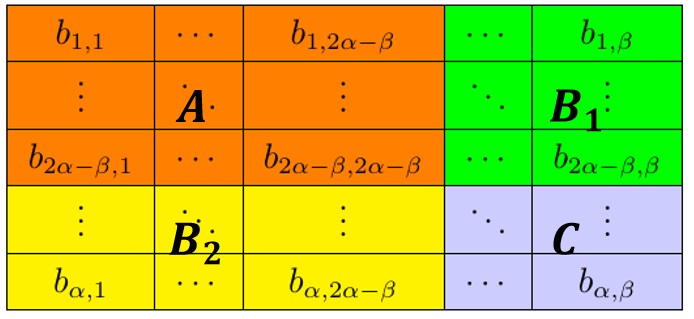}
    \caption{Sub-array allocation for $\alpha \times \beta $ array with $\alpha\leq \beta<2\alpha$.}
    \label{ABC}
\end{figure}
Note that some symbols are not changed in step three and we call them as {\em original symbols}. Since the symbols in sets $R'_{i,i}$ are uncoupled symbols, they are original symbols, where 
$i=1,2,\ldots,\alpha$. According to step three, we can see that some coupled symbols are also original symbols. 
The above three steps are our set transformation. When $\alpha = \beta$, then our set transformation is reduced to the base transformation \cite{8006804}.

\begin{example}
Consider an example of $(\alpha, \beta)=(4, 6)$.
According to step one, the 
$4\times 6$ array is divided into $2\times 2$ sub-array $A$, $2\times 4$ sub-array $B_1$, $2\times 2$ sub-array $B_2$ and $2\times 4$ sub-array $C$. According to step two, we divide the $24$ symbols in the $4\times 6$ array into $16$ sets $\{R_{i,j}\}_{i,j=1,2,3,4}$, where each of the eight sets $\{R_{i,j}\}_{i=1,2,3,4}^{j=1,2}$ has one symbol and each of the eight sets $\{R_{i,j}\}_{i=1,2,3,4}^{j=3,4}$ has two symbols. According to step three, the 16 sets are updated to be $\{R'_{i,j}\}_{i,j=1,2,3,4}$, where the set transformation on the $4\times 6$ array is shown in Fig. \ref{square}.
In Fig. \ref{square}, coupled sets are in the same color and each big box represents a sub-array and each small box represents a set.
\end{example}

\begin{figure}[htpb]
    \centering
    \includegraphics[width=0.495\textwidth]{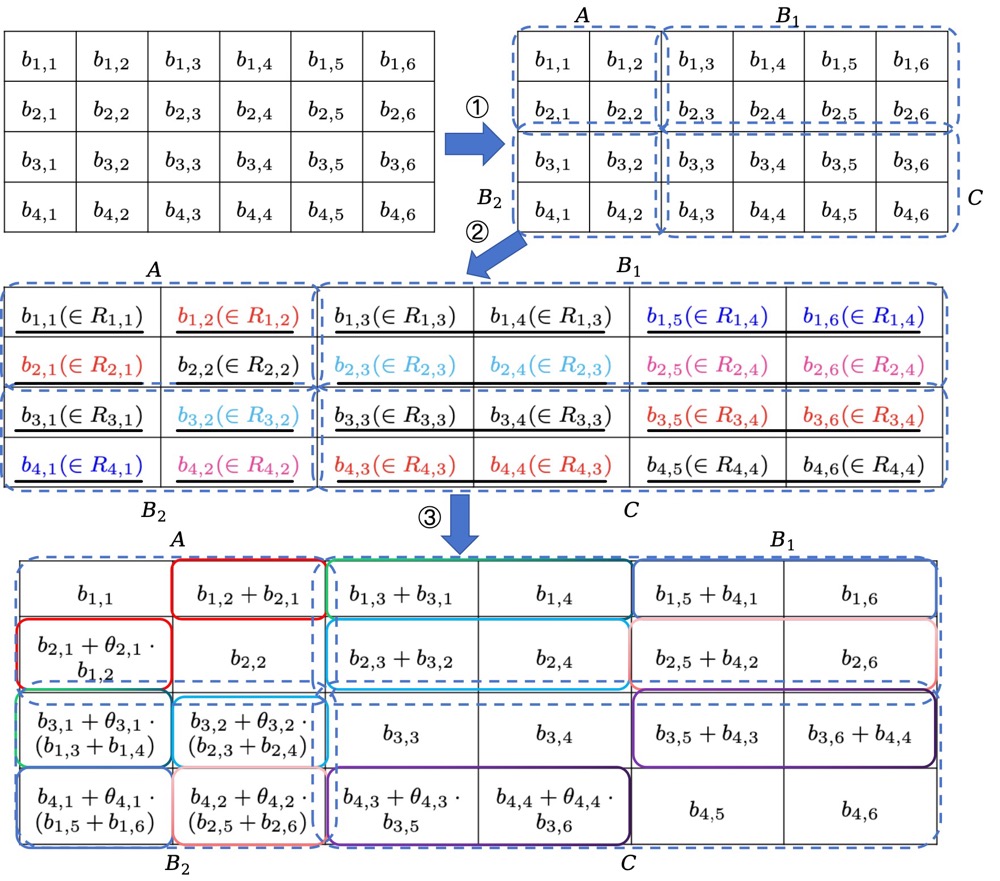}
    \caption{Set Transformation on $4\times 6$ array.}
    \label{square}
\end{figure}

The idea of our set transformation for general parameters $\alpha$ and $\beta$ with $\alpha \leq \beta$ is similar to that with $\alpha \leq \beta <2\alpha$. We summarize the idea as follows.
First, we divide the $\alpha \times \beta$ array into $(\alpha-\beta\bmod \alpha)\times(\alpha-\beta\bmod \alpha)\lfloor\frac{\beta}{\alpha}\rfloor$ sub-array $A$, $(\alpha-\beta\bmod \alpha)\times (\beta\bmod \alpha)\lceil\frac{\beta}{\alpha}\rceil$  sub-array $B_1$, $(\beta\bmod \alpha) \times (\alpha-\beta\bmod \alpha)\lfloor\frac{\beta}{\alpha}\rfloor$ sub-array $B_2$ and $(\beta\bmod \alpha) \times (\beta\bmod \alpha)\lceil\frac{\beta}{\alpha}\rceil$ sub-array $C$. Second, we divide each sub-array into many sets such that each set contains some symbols. Specifically, there are $\lfloor\frac{\beta}{\alpha}\rfloor$ symbols for each set in sub-array A and sub-array $B_2$, $\lceil\frac{\beta}{\alpha}\rceil$ symbols for each set in sub-array $B_1$ and sub-array $C$. Third, we perform the set pairwise combination for all the sets. Notice that if $\beta \geq 2\alpha$, and $\alpha \nmid \beta$, we can split it into $\lceil \frac{\beta}{\alpha} \rceil - 2$ arrays of size $\alpha \times \alpha$ and one $\alpha \times (\alpha + (\beta \bmod \alpha))$ array. If $\beta \geq 2\alpha$ and $\alpha \mid \beta$, we can split it into $\frac{\beta}{\alpha}$ arrays of size $\alpha \times \alpha$. We can show that the repair bandwidth is the same as the codes with $\beta < 2\alpha$.

\section{Set Transformed Codes via Set Transformation}
In this section, we present the construction of our set transformed codes by employing the set transformation for basis MDS codes, propose the efficient repair method for our set transformed codes, and show that our transformed codes are MDS codes.

\subsection{Set Transformed Codes}
Our set transformed codes are an $\alpha \times n$ array, where $r=n-k$ and $2 \leq \alpha \leq r$.

First, we form an $\alpha \times n$ array by creating $\alpha$ instances of $ (n, k)$ RS systematic code. In the array, each row represents an instance. Specifically, for $i = 1, 2, \ldots ,\alpha$, let 
\[(a_{i,1}, a_{i,2}, \ldots,a_{i,k},f_1(\boldsymbol{a}_i),\ldots,f_r(\boldsymbol{a}_i))\] 
be the instance in the $i$-th row, where $\boldsymbol{a}_i := (a_{i,1}, a_{i,2}, \ldots, a_{i,k})$
represents the $k$ data symbols in this instance, and $\{f_s(\boldsymbol{a}_i)\}_{s=1}^{r}$ represent the $r$ parity symbols.

Second, we divide the $\alpha\times k$ array consisting by the first $k$ columns of the $\alpha\times n$ array into $\lfloor \frac{k}{\alpha}\rfloor$ sub-arrays $\{\Phi_i\}_{i = 1,2,\ldots,\lfloor \frac{k}{\alpha}\rfloor}$. If $\alpha|k$, then each sub-array is of size $\alpha\times \alpha$. Otherwise, if $\alpha\nmid k$, then $\Phi_{\lfloor \frac{k}{\alpha}\rfloor}$ is of size $\alpha\times (k-(\lfloor \frac{k}{\alpha}\rfloor-1)\alpha))$ and each of the other sub-array is of size $\alpha\times \alpha$. Similarly, we divide the $\alpha\times r$ array consisting by the last $r$ columns into $\lfloor \frac{r}{\alpha}\rfloor$ sub-arrays. If $\alpha|r$, each sub-array is of size $\alpha\times \alpha$. Otherwise, if $\alpha\nmid r$, the last sub-array is of size $\alpha\times (r-(\lfloor \frac{r}{\alpha}\rfloor-1)\alpha))$ and each of the other sub-array is of size $\alpha\times \alpha$. Therefore, we have divided the $\alpha\times n$ array into $\lfloor \frac{k}{\alpha}\rfloor+\lfloor \frac{r}{\alpha}\rfloor$ sub-arrays, each sub-array is of size $\alpha\times \alpha$ or $\alpha\times (k-(\lfloor \frac{k}{\alpha}\rfloor-1)\alpha))$ or $\alpha\times (r-(\lfloor \frac{r}{\alpha}\rfloor-1)\alpha))$.

Third, we perform the set transformation on each of the obtained $\lfloor \frac{k}{\alpha}\rfloor+\lfloor \frac{r}{\alpha}\rfloor$ sub-arrays to obtain $\lfloor \frac{k}{\alpha}\rfloor+\lfloor \frac{r}{\alpha}\rfloor$ {\em transformed-arrays} which form the $\alpha\times n$ array. We store the $\alpha$ symbols in column $j$ in node $j$, where $j=1,2,\ldots,n$. 

The resulting codes are both called {\em set transformed codes} and we denote the codes as the {\em ST-RS$(n,k,\alpha)$} codes. 
Note that we have divided the $\alpha \times n$ array into $\lfloor \frac{k}{\alpha}\rfloor+\lfloor \frac{r}{\alpha}\rfloor$ sub-arrays in the above construction. Similarly, we can also divide the $\alpha \times n$ array into $\lfloor \frac{n}{\alpha}\rfloor$ sub-arrays to obtain our codes, as like elastic transformed codes.

For notational convenience, denote the symbol in row $i$ and column $j$ of {\em ST-RS$(n,k,\alpha)$} codes as $x_{i,j}$, where $i=1,2,\ldots,\alpha$ and $j=1,2,\ldots,n$.

\subsection{Repair Process for Single-Node Failure}

Suppose that node $t$ fails, where $t \in \{1,2,\ldots, n\}$, we need to repair the $\alpha$ symbols $x_{1,t},x_{2,t},\ldots,x_{\alpha,t}$ in the failed node $t$. Recall that we have divided the $\alpha \times n$ array into $\lfloor \frac{k}{\alpha}\rfloor+\lfloor \frac{r}{\alpha}\rfloor$ sub-arrays in the second step of our construction and have employed the set transformation for each sub-array to obtain the transformed-array. Therefore, we can obtain $\alpha^2$ sets $R'_{i,j}$ for each transformed-array, where $i,j\in\{1,2,\ldots,\alpha\}$. The $\alpha$ symbols in the failed node $t$ are located in one column of one transformed array. We label the indices of the $\lfloor \frac{k}{\alpha}\rfloor+\lfloor \frac{r}{\alpha}\rfloor$ transformed-arrays from 1 to $\lfloor \frac{k}{\alpha}\rfloor+\lfloor \frac{r}{\alpha}\rfloor$.
Without loss of generality, suppose that the $\alpha$ failed symbols are located in transformed-array $\ell$, where $1\leq \ell\leq \lfloor \frac{k}{\alpha}\rfloor+\lfloor \frac{r}{\alpha}\rfloor$.
We present the repair process of node $t$ as follows.
\begin{itemize}
    \item \textbf{Step 1 (Selecting the major row)}:   Note that in transformed-array $\ell$, we have $\alpha^2$ sets $\{R'_{i,j}\}_{i,j=1,2,\ldots,\alpha}$. According to step three of our set transformation in Section II, the set $R'_{s,s}$ contains one or two symbols that are symbols of an $(n,k)$ RS codeword. There must exist one integer $s\in\{1,2,\ldots,\alpha\}$ such that $x_{s,t}\in R'_{s,s}$.
We call row $s$ as the {\em major row} of node $t$.

\item \textbf{Step 2 (Downloading helper symbols)}: We download the following three sets of symbols: (i) set $\mathcal{S}_1$ that contains $k$ symbols in the major row, except the failed symbol and symbols coupled with the failed symbols. Note that we can always find the $k$ symbols in $\mathcal{S}_1$, since $\alpha \leq r$ and $n-\alpha\geq k$; (ii) set $\mathcal{S}_2$ that contains the symbols coupled with the $k$ symbols in set $\mathcal{S}_1$; (iii) set $\mathcal{S}_3$ that contains the symbols coupled with the failed symbols. 

\item \textbf{Step 3 (Repairing the failed symbols in a major row)}: We use the symbols in sets $\mathcal{S}_1$ and $\mathcal{S}_2$ to recover the codeword symbols of the 
$(n,k)$ RS code in the major row, i.e., row $s$.

\item \textbf{Step 4 (Recovering the symbols in the failed node $t$)}: We use the symbols in set $\mathcal{S}_3$, together with the $k$ symbols recovered in Step 3, to repair all symbols in node $t$. 
\end{itemize}

Following the above steps we can complete the single node failure repair of ST-RS $(n,k,\alpha)$, and we can use the same method to repair a single failed node for the case where the $n\times\alpha$ array is directly divided into $\lfloor\frac{n}{\alpha}\rfloor$ transformed arrays.

\begin{example}
Consider the example of $(n, k,\alpha)=(14, 10,3)$, the ST-RS$(14,10,3)$ code is an array of size $3\times 14$. The three symbols in column $j$ are stored in node $j$, where $j=1,2,\ldots,14$.
We divide the $3\times 14$ array into four transformed arrays, each of the first two transformed arrays is of size $3\times 3$ and each of the last two transformed arrays is of size $3\times 4$. 
Fig. \ref{10+4} shows the array of ST-RS$(14,10,3)$ code. 
    
Suppose that node 1 has failed, we need to repair the $3$ symbols $x_{1,1},x_{2,1},x_{3,1}$. Since node 1 is located in the first transformed array, we have $\ell=1$. According to step three of the set transformation in Section II, we have that 
\[
(x_{1,1},x_{2,1},x_{3,1})=(a_{1,1},a_{2,1}+\theta_{2,1} \cdot a_{1,2},a_{3,1}+\theta_{3,1}\cdot a_{1,3}).
\]
Since $R'_{1,1}=\{x_{1,1}\}$ and $x_{1,1}=a_{1,1}$ is a symbol of the $i$-th instance of $(14,10)$ RS code, we have that $s=1$ and row $s=1$ is the major row of node $t=1$. According to step 2 of the repair process, we can download the $k=10$ symbols in set $\mathcal{S}_1$ (the symbols with green color in Fig. \ref{10+4}), download the $5$ symbols in set $\mathcal{S}_2$ (the symbols with purple color in Fig. \ref{10+4}) and download the $2$ symbols in set $\mathcal{S}_3$ (the symbols with blue color in Fig. \ref{10+4}) to repair the erased $3$ symbols. 
According to step 3 of the repair process, we can recover the $k=10$ symbols 
$\{a_{1,j}\}_{j=1,2,\ldots,10}$
from the downloaded symbols in sets $\mathcal{S}_1$ and $\mathcal{S}_2$.
Specifically, we can compute
$a_{1,5}=(\theta_{2,4}-1)^{-1}(x_{2,4}-x_{1,5}), a_{2,4}=(\theta_{2,4}-1)^{-1}(\theta_{2,4}\cdot x_{1,5}-x_{2,4})$
from $2$ symbols $x_{2,4}=a_{2,4}+\theta_{2,4}\cdot a_{1,5}$ and $x_{1,5}=a_{2,4}+a_{1,5}\}$. 

Similarly, we can compute the $10$ symbols 
$\{a_{1,i}\}_{i=4}^{10}\cup\{f_1(\boldsymbol{a}_1), f_2(\boldsymbol{a}_1),f_4(\boldsymbol{a}_1)\}$ from the symbols in sets $\mathcal{S}_1$ and $\mathcal{S}_2$, and further recover the $3$ symbols $a_{1,1},a_{1,2},a_{1,3}$.
Together with the symbols in set $\mathcal{S}_3$ and $a_{1,2},a_{1,3}$, we can recover the failed symbols $x_{2,1}=a_{2,1}+\theta_{2,1}\cdot a_{1,2}=(x_{1,2}-a_{1,2})+\theta_{2,1}\cdot a_{1,2},x_{3,1}=a_{3,1}+\theta_{3,1}\cdot a_{1,3}=(x_{1,3}-a_{1,3})+\theta_{3,1}\cdot a_{1,3}.$


We can count that the repair bandwidth of node 1 is $17$ symbols $(56.7\%)$. Recall that the repair bandwidth of the node of elastic transformed code with the same parameter is $20$ symbols $(66.7\%)$. Our codes have a smaller repair bandwidth than that of elastic transformed code.

    \begin{figure}[H]
    \centering
    \includegraphics[width=0.50\textwidth]{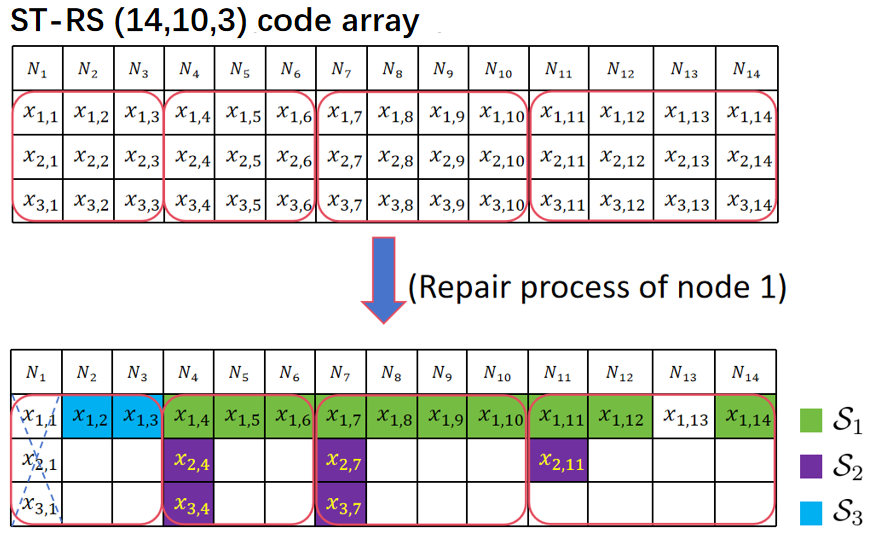}
    \caption{ST-RS$(14, 10,3)$ code with two square sub-arrays and two non-square sub-arrays.}
    \label{10+4}
\end{figure}

\end{example}

\subsection{Analysis of Repair Bandwidth}
In the following, we analyze the lower bound of repair bandwidth of our codes. 
Suppose that we have divided the $\alpha \times n$ array into $\lfloor \frac{n}{\alpha}\rfloor$ sub-arrays in our codes.

\begin{theorem}\label{repairbandwidth} 
The repair bandwidth $\gamma$ of our ST-RS$(n,k,\alpha)$ codes is lower bounded by:
    \begin{equation}\nonumber
        \begin{cases}
            k + \alpha - 1 & k \leq \lfloor\frac{n}{\alpha}\rfloor - 1 + (n\bmod \alpha),\\
            2k -\lfloor\frac{n}{\alpha}\rfloor - (n\bmod \alpha) +\alpha & k > \lfloor\frac{n}{\alpha}\rfloor - 1 + (n\bmod \alpha). \\
        \end{cases}
    \end{equation}
\end{theorem}

\begin{proof}
According to the repair process in Section III.B, the repair bandwidth is the total number of symbols in sets $\mathcal{S}_1$,  $\mathcal{S}_2$, and  $\mathcal{S}_3$.
It is easy to see that there are $k$ symbols in set $\mathcal{S}_1$. 

Next, we consider the number of symbols in set $\mathcal{S}_2$. We first claim that (i)
there exists one non-coupled symbol which is in $R'_{i,i}$ and $\alpha - 1$ coupled symbols in row $i$ of an $\alpha \times \alpha$ square transformed-array; (ii) there exists at most two non-coupled symbols which are in $R'_{i,i}$, at most $\beta-\alpha-1$ original symbols and at least coupled non-original symbols in row $i$ of an $\alpha \times \beta$ non-square transformed-array, where $i=1,2,\ldots,\alpha$. 
More original symbols mean lower repair bandwidth when selecting symbols in set $\mathcal{S}_2$.
There are at most $\lfloor\frac{n}{\alpha}\rfloor - 1 + (n\bmod\alpha)$ original symbols in each row. 
If $k \leq \lfloor\frac{n}{\alpha}\rfloor - 1 + (n\bmod\alpha)$, then we can select all the $k$ symbols in set $\mathcal{S}_1$ to be original symbols and $|\mathcal{S}_2|=0$. Otherwise, if $k >\lfloor\frac{n}{\alpha}\rfloor - 1 + (n\bmod\alpha)$, we can select $\lfloor\frac{n}{\alpha}\rfloor - 1$ non-coupled symbols, at most $n\bmod \alpha$ original symbols and at least $k-(\lfloor\frac{n}{\alpha}\rfloor - 1 + (n\bmod\alpha))$ or at most $k-(\lfloor\frac{n}{\alpha}\rfloor - 1)$
coupled symbols in set $\mathcal{S}_1$.
Therefore, the number of symbols in set $\mathcal{S}_2$ is at least $k-(\lfloor\frac{n}{\alpha}\rfloor - 1 + (n\bmod\alpha))$.

There are at least $\alpha - 1$ symbols coupled with the $\alpha$ failed symbols. Therefore, the number of symbols in set $\mathcal{S}_3$ is at least $\alpha-1$. The theorem is proved.
\end{proof}


\subsection{The MDS Property}
The next theorem shows the filed size of our ST-RS$(n,k,\alpha)$ codes to maintain the MDS property.

\begin{theorem}
    \label{MDS}
If the field size is larger than
    \begin{equation}
        \binom{n - 1}{k - 1} - \binom{\lceil \frac{n}{\alpha} \rceil - 1}{\lceil \frac{k}{\alpha} \rceil - 1} ,
    \end{equation}
then we can always find the values $\theta_{i,j}$ such that our codes are MDS codes.
\end{theorem}

\begin{proof}
It is sufficient to show that any $k$ out of the $n$ nodes can recover all the data symbols.
Suppose that $r$ nodes are erased, there are $\binom{n}{k}$ cases in selecting any $k$ out of the $n$ nodes. For each case, the determinant of the corresponding coefficient matrix is a polynomial of 
variables $\theta_{i,j}$, where $i = 1,2,\ldots,\alpha$ and $j = 1,2,\ldots,n$. The variable $\theta_{i,j}$ is used in at most one entry in one square set transformation array, and at most two entries in a non-square set transformation array. Thus the degree of  $\theta_{i,j}$ of the determinant polynomial is at most one. We want to find the values for variables $\theta_{i,j}$ such that the evaluation of the determinants multiplication of all $\binom{n}{k}$ cases is non-zero. The degree of variable $\theta_{i,j}$ of all $\binom{n}{k}$ determinants multiplication is at least $\binom{n-1}{k-1}$. By \cite[Theorem 1.2]{alon_1999}, if the field size $q$ is larger than the degree of each $\theta_{i,j}$ in the determinants multiplication, there exists at least one value for each variable $\theta_{i,j}$ such that the evaluation of the determinants multiplication is non-zero. Therefore, our codes are MDS codes if the field size is $\binom{n-1}{k-1}$.

Note that we can always retrieve all the data symbols from some cases of $k$ nodes. If the $k$ nodes are chosen as all the columns in some transformed arrays, then we can compute all the data symbols \cite[Theorem 1]{8410002}.
There are $\binom{\lceil \frac{n}{\alpha} \rceil - 1}{\lceil \frac{k}{\alpha} \rceil - 1}$
such special cases. We only need to make sure that the evaluation of the multiplication of $\binom{n-1}{k-1}-\binom{\lceil \frac{n}{\alpha} \rceil - 1}{\lceil \frac{k}{\alpha} \rceil - 1}$ determinants is non-zero, and the theorem is proved.
\end{proof}


\section{Comparison}

In this section, we evaluate the performance of our ST-RS$(n,k,\alpha)$ codes and the existing MDS codes in terms of field size, sub-packetization level and repair bandwidth.




Table \ref{tab:codes_and_fields} shows the field size and sub-packetization level of our codes and the existing related MDS codes, including HTEC \cite{8025778}, $\mathcal{C}_1$ \cite{2023wk1}, $\mathcal{C}_2$ \cite{7084873}, $\mathcal{C}_3$ \cite{hou2019multilayer}, and ET-RS codes \cite{10228984} that have low repair bandwidth and small sub-packetization level. We can see from Table \ref{tab:codes_and_fields} that our codes have lower field size than all the other evaluated codes, and have a flexible sub-packetization level like ET-RS codes.
Note that we can recursively apply our set transformation for RS codes many times to achieve any sub-packetization level like ET-RS codes \cite{10228984}.

Define the average repair bandwidth ratio as the ratio of the average repair bandwidth of all $n$ nodes to the number of data symbols.
In the following, we first show that the repair bandwidth of our ST-RS$(n,k,\alpha)$ codes is strictly lower than that of ET-RS$(n,k,\alpha)$ codes \cite{10228984} (refer to Theorem \ref{repairbandwidth_theorem}).
Then we evaluate the average repair bandwidth ratio for some typical parameters of our ST-RS$(n,k,\alpha)$ codes and ET-RS codes (refer to Table \ref{tab:my_label}, these data are exactly the repair bandwidth of the repair methods).

Next lemma shows that our codes have strictly less repair bandwidth than ET-RS codes for some nodes.

\begin{table}[H]
    \centering
    \resizebox{0.88\linewidth}{1.9cm}{
    \begin{tabular}{|c|c|c|}
    \hline
    Codes & Field Size & Sub-packetization \\
    \hline
    $\mathcal{C}_2$\cite{7084873} &  $\binom{n}{k}(n-k)^{\alpha + 1}$ & $r^{\frac{k}{r}}$\\
    \hline
    $\mathcal{C}_3$\cite{hou2019multilayer} & $\binom{n}{k}\frac{\alpha(\alpha-1)}{2}\lfloor\frac{k}{\alpha}\rfloor$ & $(d-k+1)^{\lfloor\frac{n}{(d-k+1)\eta}\rfloor}$ \\
    \hline
    HTEC\cite{8025778} & $\binom{n}{k}(n-k)\alpha$ & $2 \leq \alpha\leq r^{\lceil\frac{k}{r}\rceil}$\\
    \hline
    $\mathcal{C}_1$ \cite{2023wk1} & $\binom{n-1}{k-1}+2$& $2\leq \alpha\leq r$\\
    \hline
    ET-RS \cite{10228984} & $2\Big(\binom{n - 1}{k - 1} - \binom{\lceil \frac{n}{\alpha*} \rceil - 1}{\lceil \frac{k}{\alpha*} \rceil - 1}\Big)$& $2\leq \alpha\leq r^{\lfloor\frac{n}{r}\rfloor}$\\
    \hline
    \textbf{Our codes} & $\binom{n - 1}{k - 1} - \binom{\lceil \frac{n}{\alpha*} \rceil - 1}{\lceil \frac{k}{\alpha*} \rceil - 1}$ &$2 \leq \alpha\leq r^{\lfloor\frac{n}{r}\rfloor}$\\
    \hline
    \end{tabular}}
    \caption{Evaluation for our codes and related codes.}
    \label{tab:codes_and_fields}
\end{table}


\begin{lemma}\label{bandwidthproof1}
Let $\gamma_{\text{ST-RS}}$ be the repair bandwidth of single-node failure in the first $n-2(n\bmod\alpha)$ nodes of our codes, and $\gamma_{\text{ET-RS}}$ be the lower bound on the repair bandwidth of single-node failure in the first $n-2(n\bmod\alpha)$ nodes of ET-RS$(n,k,\alpha)$. We have that $ \gamma_{ET-RS} - \gamma_{ST-RS} \geq 0.$ Moreover, there exist $(\alpha - (n\bmod \alpha))\lfloor\frac{n}{\alpha}\rfloor$ nodes such that $\gamma_{ET-RS}-\gamma_{ST-RS} \geq n\bmod\alpha.$
\end{lemma}
\begin{proof}
    See Appendix \ref{app1}.
\end{proof}

The next theorem shows that our codes have a lower average repair bandwidth ratio than ET-RS codes for high-code-rate parameters.

\begin{theorem}\label{repairbandwidth_theorem}
Suppose that $n / k > 0.5$, the average repair bandwidth ratio $\bar{\gamma}_{ST-RS}$ of our ST-RS$(n, k, \alpha)$   is lower than that of ET-RS$(n, k, \alpha)$.
\end{theorem}
\begin{proof}
    See Appendix \ref{app2}.
\end{proof}

Table \ref{tab:my_label} shows the average repair bandwidth ratio of our ST-RS$(n,k,\alpha)$ codes, ET-RS codes and HTEC codes, under some typical parameters. Note that the values are exactly the repair bandwidth in the table.
It can be seen from Table \ref{tab:my_label} that ST-RS codes have lower repair bandwidth than the other two codes for all the evaluated parameters.
\begin{table}[H]
    \centering
    \renewcommand{\arraystretch}{1.5}
    \begin{tabular}{|c|c|c|c|c|}
\hline
 $(n,k,\alpha)$        & \textbf{ET-RS} &   \textbf{HTEC} & \textbf{ST-RS} & \textbf{Cut-set bound}\\
\hline
    $(10, 7, 3)$    & $72.3\%$ & $68.5\%$ & \textbf{65.7\%} & $42.8\%$ \\
\hline
     $(14, 10, 4)$     & $55.3\%$  & $ 60.1\%$ & \textbf{51.7\%} & $32.5\%$\\
\hline
    $(17, 13, 4)$     &  $54.2\%$  & $57.2\%$ & \textbf{49.7\%} &$30.7\%$ \\
\hline
     $(22, 18, 4)$    & $50.1\%$ & $54.1\%$ & \textbf{48.1\%} & $29.1\%$ \\
\hline
    $(29, 25, 4)$     & $49.0\%$ &  $51.5\%$ & \textbf{46.8\%} &$28.0\%$ \\
\hline
    \end{tabular}
    \caption{Average repair bandwidth ratio of three codes.}
    \label{tab:my_label}
\end{table}

\section{Conclusion}
In this paper, we propose set transformation that can transform base MDS code into
set transformed MDS code with lower repair bandwidth.
Compared with existing elastic transformed codes, our set transformed codes require a smaller field size and lower repair bandwidth. How to design a more general transformation structure to further reduce the repair bandwidth is one of our future works.

\ifCLASSOPTIONcaptionsoff
\newpage
\fi

\bibliography{ck}
\newpage

\begin{appendices}
      \section{Proof of Lemma \ref{bandwidthproof1}}\label{app1}
\begin{proof}
    Referring to Theorem \ref{repairbandwidth} and \cite[Theorem 2]{10228984}, the repair process of both set transformation and elastic transformation is divided into three sets of downloading, we note that the three sets of set transformation are sets $\mathcal{S}^{T}_1,\mathcal{S}^{T}_2,\mathcal{S}^T_3$, and the 
    three sets of elastic transformation are sets $\mathcal{S}^{E}_1, \mathcal{S}^{E}_2,\mathcal{S}^E_3$.

    Recall that according to the similar steps of Section III.B, the first $\alpha(\lfloor\frac{n}{\alpha}\rfloor - 1)$ nodes form $\lfloor\frac{n}{\alpha}\rfloor - 1$ $\alpha \times \alpha$ arrays, and the remaining $\alpha + n\bmod\alpha$ nodes form an $\alpha \times (\alpha + n\bmod\alpha)$ array to make transformation respectively.

    Suppose node $t$ failed and $t \in \{1,2,\ldots, n - 2(n\bmod\alpha)\}$, to recover node $t$, we discuss the three sets in the repair process in turn. 
    
    The first sets $\mathcal{S}_1^T$ and $\mathcal{S}_2^E$ are discussed first, it is easy to see that $|\mathcal{S}_1^T| = |\mathcal{S}_1^E| = k$;

    Then for the last sets, if node $t$ is not in the last $2(n\bmod\alpha)$ nodes, the symbols in node $t$ must be in the sub-array $A$ and $B_2$. Every symbol in the sub-array $A$ and $B_2$ can only download at most one coupled symbol to recover, so $|\mathcal{S}_1^T| = \alpha - 1$. However, $|\mathcal{S}_1^E| \geq \alpha - 1$.

    For the second sets, recall that it includes the coupled symbols for the $k$ symbols in the first set. In set $\mathcal{S}_1^{T}$, there are $\lfloor\frac{n}{\alpha}\rfloor - 1 + (n\bmod\alpha)$ symbols that do not need to download coupled symbols, at most $n\bmod\alpha$ symbols that need to download two coupled symbols and the rest of symbols that need to download one coupled symbols. In set $\mathcal{S}_1^{E}$, there are $\lfloor\frac{n}{\alpha}\rfloor - 1$ symbols that do not need to download coupled symbols, up to $n\bmod\alpha$ symbols that need to download two coupled symbols, and the rest of symbols that need to download one coupled symbols. We select $k$ symbols with as few coupled symbols as possible. So there must be $|\mathcal{S}^{T}_2| \leq |\mathcal{S}^{E}_2|$. Finally, 
    \begin{eqnarray*}
        \gamma_{ET-RS} - \gamma_{ST-RS} &=&\sum_{i=1}^{3}(|\mathcal{S}^{E}_i| - |\mathcal{S}^{T}_i|)\\
        &\geq&|\mathcal{S}^{E}_2| - |\mathcal{S}^{T}_2|\geq 0.
    \end{eqnarray*}    

    If node $t$ in the first $\alpha - (n \bmod \alpha)$ nodes of every set transformation array (totally $(\alpha - (n \bmod \alpha))\lfloor\frac{n}{\alpha}\rfloor$ nodes), there must be at $n\bmod\alpha$ coupled but original symbols in the major row because the second symbol of the sets in sub-array $B_1$ is original. At this point, $$|\mathcal{S}^{T}_2| = k-(\lfloor\frac{n}{\alpha}\rfloor - 1 + (n\bmod\alpha))$$ and $$|\mathcal{S}^{E}_2| \geq k-(\lfloor\frac{n}{\alpha}\rfloor - 1 ).$$ Finally, we can get
    \begin{eqnarray*}
        \gamma_{ET-RS} - \gamma_{ST-RS} &=& (\sum_{i=1}^{3}(|\mathcal{S}^{E}_i| - |\mathcal{S}^{T}_i|)\\
        &\geq& |\mathcal{S}^{E}_2| - |\mathcal{S}^{T}_2|\\
        &\geq& n\bmod\alpha.
    \end{eqnarray*}
    
\end{proof}

\section{Proof of Theorem \ref{repairbandwidth_theorem}}\label{app2}
\begin{proof}
    From Lemma \ref{bandwidthproof1}, we can know that for the first $n - 2(n\bmod \alpha)$ nodes, ST-RS$(n, k, \alpha)$ has at least $(\alpha - n\bmod\alpha)\lfloor\frac{n}{\alpha}\rfloor (n\bmod\alpha) $ symbols repair bandwidth less than ET-RS$(n, k, \alpha)$. Recall that the average repair bandwidth ratio is defined as the ratio of the sum of the repair bandwidth of all nodes and the number of nodes times $k\alpha$, below it is only necessary to show that the total repair bandwidth of the $2(n\bmod \alpha)$ remaining nodes ST-RS$(n, k, \alpha)$ is at most $(\alpha - n\bmod\alpha)\lfloor\frac{n}{\alpha}\rfloor (n\bmod\alpha) $ symbols more than that of ET-RS$(n, k, \alpha)$.

    Like the proof process of Lemma \ref{bandwidthproof1}, suppose node $t$ failed, to recover node $t$, we discuss the three sets in the repair process in turn. 
    
    The first sets $\mathcal{S}_1^T$ and $\mathcal{S}_2^E$ are discussed first, it is easy to see that $|\mathcal{S}_1^T| = |\mathcal{S}_1^E| = k$;

    For the second sets, like the similar proof in Lemma \ref{bandwidthproof1}, in set $\mathcal{S}_1^{T}$, there are $\lfloor\frac{n}{\alpha}\rfloor - 1 + (n\bmod\alpha)$ symbols that do not need to download coupled symbols and the rest of symbols that need to download one coupled symbols. In set $\mathcal{S}_1^{E}$, there are $\lfloor\frac{n}{\alpha}\rfloor - 1$ symbols that do not need to download coupled symbols and the rest of the symbols that need to download one coupled symbols. We select $k$ elements with as few coupled symbols as possible. So there must be $|\mathcal{S}^{T}_2| \leq |\mathcal{S}^{E}_2|$. Finally, 
    \begin{eqnarray*}
        \gamma_{ET-RS} - \gamma_{ST-RS} &=& \sum_{i=1}^{3}(|\mathcal{S}^{E}_i| - |\mathcal{S}^{T}_i|)\\
        &\geq& |\mathcal{S}^{E}_2| - |\mathcal{S}^{T}_2| \geq 0.
    \end{eqnarray*}

    Finally, for the last sets, given that each pair of symbols in sub-array \(B_1\) corresponds to a single symbol in ST-RS\((n, k, \alpha)\), we possibly need to download two symbols to recover one symbol in sub-array $B_1$, and other symbols only need to download one symbol. Recall that the number of symbols in sub-array $B_1$ is $2(n\bmod\alpha)(\alpha-n\bmod\alpha)$. On the other hand, one lost symbol in ET-RS$(n, k, \alpha)$ at least needs to download one symbol to recover it. In total, for ST-RS$(n, k, \alpha)$, we at most download $2(n\bmod\alpha)(\alpha-n\bmod\alpha)$ symbols more than that of ET-RS$(n, k, \alpha)$ for the last set.

    In order to have a lower repair bandwidth ratio, we wish that $$(\alpha - n\bmod\alpha)\lfloor\frac{n}{\alpha}\rfloor (n\bmod\alpha) - 2(n\bmod\alpha)(\alpha-n\bmod\alpha)$$ $$\geq 0.$$
    then 
    $$(\lfloor\frac{n}{\alpha}\rfloor- 2)(n\bmod\alpha)(\alpha-n\bmod\alpha)\geq 0.$$
    If $n \bmod \alpha = 0$, then the inequality is constant, and it may be assumed that $n \bmod \alpha > 0$, while we must have $\alpha > n \bmod \alpha$.
    Simplify that we can get $$\lfloor\frac{n}{\alpha}\rfloor - 2  \geq 0,$$
    we only need $$\lfloor\frac{n}{\alpha}\rfloor \geq 2.$$ Because of $2k > n > 2r \geq 2\alpha$, the conclusion is clearly valid.
\end{proof}
\end{appendices}

\end{document}